\documentclass[11pt,reqno]{amsart}%

\usepackage{amsthm,amsmath,amsfonts,amssymb,amsxtra,appendix,bookmark,euscript,delarray,dsfont,color,mathrsfs}
\usepackage[utf8]{inputenc}
\usepackage{bbm}
\usepackage{hyperref}

\newtheorem{theorem}{Theorem}%[section]

\newtheorem{lemma}[theorem]{Lemma}
\newtheorem{corollary}[theorem]{Corollary}

\theoremstyle{definition}

\newtheorem{assumption}[theorem]{Assumption}

\theoremstyle{remark}

\newtheorem{remark}[theorem]{Remark}

\newcommand{\1}{\mathbbm{1}}

\renewcommand{\epsilon}{\varepsilon}

\renewcommand{\phi}{\varphi}
\newcommand{\R}{\mathbb{R}}

\newcommand{\Sph}{\mathbb{S}}

\newcommand{\cL}{\mathcal{L}}
\newcommand{\lesim}{\lesssim} 
\newcommand{\gesim}{\gtrsim}

\DeclareMathOperator{\dom}{dom}

\DeclareMathOperator{\supp}{supp}

\title[CLR type inequalities for Hardy--Schr\"odinger operator]{Cwikel--Lieb--Rozenblum type inequalities\\ for Hardy--Schr\"odinger operator}

\author[G.K. Duong]{Giao Ky Duong}
\address[G.K. Duong]{Department of Mathematics, LMU Munich, and Munich Center for Quantum Science and Technology (MCQST), Germany} 
\email{duong@math.lmu.de}

\author[R.L. Frank]{Rupert L. Frank}
\address[R.L. Frank]{Department of Mathematics, LMU Munich, and Munich Center for Quantum Science and Technology (MCQST), Germany, and Department of Mathematics, Caltech, USA}
\email{r.frank@lmu.de}

\author[T.M.T. Le]{Thi Minh Thao Le}
\address[T.M.T. Le]{Department of Mathematics and Statistics, Masaryk University, Brno, Czech Republic}
\email{tmtle@math.muni.cz}

\author[P.T. Nam]{Phan Th\`anh Nam}
\address[P.T. Nam]{Department of Mathematics, LMU Munich, and Munich Center for Quantum Science and Technology (MCQST), Germany} 
\email{nam@math.lmu.de}

\author[P.T. Nguyen]{Phuoc-Tai Nguyen}
\address[P.T. Nguyen]{Department of Mathematics and Statistics, Masaryk University, Brno, Czech Republic}
\email{ptnguyen@math.muni.cz}

\begin{document}

\maketitle

\begin{abstract} We prove a Cwikel--Lieb--Rozenblum type inequality for the number of negative eigenvalues of the Hardy--Schr\"odinger operator $-\Delta - (d-2)^2/(4|x|^2) -W(x)$ on $L^2(\mathbb{R}^d)$. The bound is given in terms of a weighted $L^{d/2}-$norm of $W$ which is sharp in both large and small coupling regimes. We also obtain a similar bound for the fractional Laplacian.

%\medskip \noindent
%{\sc Résumé.} Nous prouvons une inégalité de type Cwikel--Lieb--Rozenblum pour le nombre de valeurs propres négatives de l'opérateur de Hardy--Schrödinger $-\Delta - (d-2)^2/(4|x|^2) - W(x)$ sur $L^2(\mathbb{R}^d)$. La borne est donnée en termes d'une norme pondérée $L^{d/2}$ de $W$, qui est optimale à la fois dans les régimes de couplage fort et faible. Nous obtenons également une borne similaire pour le Laplacien fractionnaire.
%
%\medskip \noindent
%{\sc Keywords:} Schr\"odinger operator, semiclassical estimates, Cwikel--Lieb--Rozenblum inequality, singular potentials
%
%\medskip \noindent
%{\sc Mathematics Subject Classification 2020:} 81Q20

\end{abstract}

\section{Introduction and Main Results}

The celebrated Cwikel--Lieb--Rozenblum (CLR) inequality \cite{Cwikel-77,Lieb-76,Rozenblum-76} states that for all dimensions $d\ge 3$, the number of negative eigenvalues of the Schrödinger operator $-\Delta  -V$ on $L^2(\mathbb{R}^d)$, with a real-valued potential $V\in L^{d/2}(\R^d)$, satisfies 
\begin{align}\label{eq:CLR}
N(0, -\Delta-V) \lesssim_d \int_{\R^d} V(x)_+^{d/2} {\rm d} x
\end{align}
where $V(x)_+=\max(V(x),0)$. Here the notation $\lesssim_d$ means that the implicit constant on the right hand side depends only on the dimension $d$. In particular, since $N(0, -\Delta-V)$ is always an integer, \eqref{eq:CLR} implies that $N(0,-\Delta -V)=0$ if $\|V_+\|_{L^{d/2}(\R^d)}$ is small enough, which can be deduced from Sobolev's inequality
\begin{align}\label{eq:Sobolev}
\int_{\R^d} |\nabla u(x)|^2 {\rm d} x \ge S_d \left( \int_{\R^d} |u(x)|^{\frac {2d} {d-2}} {\rm d}x  \right)^{\frac {d-2} d}
\end{align}
via the duality argument
\begin{align}  \label{eq:duality-argument}
\inf_{\|V_+\|_{L^{d/2}(\R^d)} \le S_d } \langle u, (-\Delta -V) u\rangle =  \| \nabla u\|_{L^2(\R^d)}^2 - S_d  \|u\|_{L^{\frac {2d} {d-2}} (\R^d)} ^2  \ge 0. 
\end{align}
However, the CLR inequality \eqref{eq:CLR} is much deeper than Sobolev's inequality since it captures correctly the semiclassical behavior which is usually described by Weyl's law in the large coupling regime 
\begin{align}\label{eq:CLR-phase-space}
N(0, -\Delta-\lambda V) & = \frac 1 {(2\pi)^{d}} |\{ (p,x) \in \R^d \times \R^d: |p|^2 - \lambda V(x) <0\} |+ o(\lambda^{d/2})_{\lambda\to \infty} \nonumber\\
&=\frac{ |B| }{(2\pi)^{d}} \int_{\R^d} (\lambda V(x))_+^{d/2} {\rm d} x + o(\lambda^{d/2})_{\lambda\to \infty}
\end{align}
where $|B|$ is the volume of the unit ball $B=\{x \in \R^d: |x|<1\}$. We refer to  \cite{FLW2023} for a textbook introduction to \eqref{eq:CLR}, \eqref{eq:CLR-phase-space} and related estimates. 

In the present paper, we are interested in potentials  of the form
$$
V(x)=  \frac{(d-2)^2}{4|x|^2}  + W(x),\quad W\in L^{\frac{d}{2}}(\R^d)
$$
where the singular part comes from Hardy's inequality 
\begin{align}\label{eq:Hardy}
\mathcal L = -\Delta - \frac{(d-2)^2}{4|x|^2} \ge 0 
\qquad\text{on}\ L^2(\R^d).
\end{align}

It was proved in \cite{EkFr} that if the Hardy--Schr\"odinger operator $\mathcal{L}-W(x)$ has negative eigenvalues $\{E_n\}_{n\ge 1}$, then 
\begin{align} \label{eq:intro-HLT}
\sum_{n\ge 1} \bigl|E_n\bigr|^\gamma
\lesssim_{\gamma,d}  \int_{\R^d} W(x)_+^{\gamma+d/2} {\rm d} x
\end{align}
for all $d\ge 3$ and $\gamma>0$.  The Hardy--Lieb--Thirring inequality \eqref{eq:intro-HLT} is an improvement over standard Lieb-Thirring inequalities   \cite{LieThi-75,LieThi-76} concerning similar estimates for the Schr\"odinger operator $-\Delta-W$. 

On the other hand, it is well-known that \eqref{eq:intro-HLT} does not hold for $\gamma=0$ \cite{EkFr}. In fact, even the corresponding Sobolev inequality does not hold with $\|u\|_{\dot H^1}^2$ replaced by $\langle u, \mathcal{L}u\rangle$. However, there is a remarkable replacement for the Sobolev inequality in the restricted case where $\R^d$ is replaced by the unit ball $B$. To be precise, it was proved by Filippas--Tertikas in \cite{FTJFA2002} (see also Musina's remarks in \cite{MJFA09}) that 
	\begin{equation}	\label{eq:Hardy-Sobolev} 
	\langle u, \cL u\rangle \gesim_d \left (  \int_{B} \frac{ |u(x)|^{\frac{2d}{d-2}}}{(1+ |\ln |x| | )^{1+\frac{d}{d-2}}} {\rm d}x \right)^{\frac {d-2}{d}} ,\quad u\in C_c^\infty(B)
	\end{equation}
where the power of the logarithmic weight is optimal. %Here $B$ denotes the unit ball in $\R^d$. 
By a duality argument similar to \eqref{eq:duality-argument}, the Hardy--Sobolev inequality \eqref{eq:Hardy-Sobolev} is equivalent to the fact that $\mathcal{L}-W \ge 0$ on $L^2(B)$, as quadratic forms with Dirichlet boundary conditions, if 
$$
 \int_{B} W(x)_+^{\frac d2} (1+|\ln |x||)^{d-1} {\rm d}x\,
$$
is sufficiently small. 

Our first new result is an extension of the above Hardy--Sobolev inequality concerning the number of negative eigenvalues of $\mathcal{L}-W$ on $L^2(\R^d)$. In particular, it allows to extend \eqref{eq:Hardy-Sobolev} to the whole $\R^d$. 

\begin{theorem}[CLR type bound for Hardy--Schr\"odinger operator]\label{thm:CLR} For every dimension $d\ge 3$, there exists a constant $C_d>0$ independent of the real-valued potential $W$ such that
	\begin{align}\label{eq:CLR-Hardy}
	N(0,\mathcal L - W) \le 1 + C_d \int_{\R^d} W(x)_+^{\frac d2} (1+|\ln |x||)^{d-1} {\rm d}x\,.
	\end{align}
\end{theorem}

Here when the right-hand side of \eqref{eq:CLR-Hardy} is finite, $\cL-W$ is bounded from below with the core domain $C_c^\infty(\R^d\backslash\{0\})$ and extended to be a self-adjoint operator by Friedrichs' method.

\begin{remark}
	The number 1 on the right-hand side of \eqref{eq:CLR-Hardy} cannot be removed. This follows from the fact that the operator $\mathcal L-\lambda W$ has a negative eigenvalue for all $\lambda>0$, whenever $W\geq 0$ and $W\not\equiv 0$; see, e.g., \cite[Remark 1.4 and Proposition 3.2]{EkFr}. In this situation, our bound \eqref{eq:CLR-Hardy} implies that $\mathcal L-\lambda W$ has exactly one negative eigenvalue for $\lambda>0$ small, and hence it  is optimal in the small coupling regime. Our bound also captures the optimal $\lambda^{d/2}$-behavior of $N(0,\mathcal L - \lambda W)$ for $\lambda>0$ large.  
\end{remark}

\begin{remark}The number 1 on the right-hand side can be removed if we restrict the consideration to any ball $B_R=\{x \in \R^d: |x|<R\} \subset \R^d$ with Dirichlet boundary conditions.  More precisely, it follows from our proof that if we restrict the quadratic form of $\mathcal L - W$ on the ball $B_R$, then
	\begin{align}\label{eq:CLR-Hardy-BR}
	N(0,\mathcal L - W) \le C_d \int_{B_R} W(x)_+^{\frac d2} (1+|\ln |x/R||)^{d-1} {\rm d}x \,.
	\end{align}
	We note that the logarithmic weight in \eqref{eq:CLR-Hardy-BR} is a consequence of the presence of the critical Hardy potential $-(d-2)^2/(4|x|^2)$, which is singular at the origin. If instead we use the critical Hardy potential $-1/(4(R-|x|)^2)$, which is singular at the boundary, the inequality holds without the logarithmic weight, as shown in \cite{FL12}. 
\end{remark}

In the proof of Theorem \ref{thm:CLR}, we will use an improvement of Hardy's inequality on the orthogonal complement of radial functions, where the singular potential $(d-2)^2/(4|x|^2)$ is not critical, and then restrict the consideration to radial functions. This strategy has been used extensively in the literature; two examples  (not necessarily the earliest) are the paper by Solomyak \cite{So} and by Birman--Laptev \cite{BiLa}. These ingredients also appear in the proof of Hardy--Lieb--Thirring inequalities \eqref{eq:intro-HLT} in \cite{EkFr} as well as the Hardy--Sobolev inequality \eqref{eq:Hardy-Sobolev} in \cite{FTJFA2002,MJFA09}. The original ingredient in our paper is the treatment on the subspace of radial functions, for which we do not know a precedent. In particular, we will prove the following Strauss type inequality, which is of independent interest. 

\begin{lemma}[Strauss type estimate for radial functions]\label{lem:radial} Consider the operator $\cL_B$ defined by the quadratic form in \eqref{eq:Hardy} restricted to $L^2(B)$ with Dirichlet boundary conditions.  Then for all radial functions $\{u_n\}_{n\ge 1}$ in the quadratic form domain of $\mathcal{L}_B$ satisfying 
\begin{align}\label{eq:Pauli}
\sum_{n\ge 1}  |\sqrt{\mathcal{L}_B} u_n \rangle \langle \sqrt{\mathcal{L}_B} u_n| \le 1\quad \text{ in }L^2(B),
\end{align}
we have the pointwise estimate 
	\begin{align}\label{eq:pointwise-Bessel}
	\rho(x):= \sum_{n\ge 1} |u_n(x)|^2 \le  \frac{C_d}{|x|^{d-2}} (1+|\ln |x||) \quad \text{ for a.e. } x \in B. 
	\end{align}
	\end{lemma}
	
	\begin{remark} The bound \eqref{eq:pointwise-Bessel} is reminiscent of Strauss' pointwise decay $|u(x)|^2 \lesim_d |x|^{-(d-1)}$ of a single radial function in $H^1(\R^d)$  \cite[Lemma 1]{St77}. Our proof strategy of Lemma \ref{lem:radial} also allows to show that if radial functions $\{v_n\}_{n\ge 1} \subset H_0^1(B)$ satisfy the orthogonality \eqref{eq:Pauli} with  $\mathcal{L}_B$ replaced by the usual Dirichlet Laplacian $-\Delta_B$, then $\sum_{n\ge 1}|v_n(x)|^2 \lesim_d |x|^{-(d-2)}$ for $d\ge 3$ (see Remark \ref{rmk:optimality}). In contrast, \eqref{eq:pointwise-Bessel} is slightly worse than the latter bound since it contains a logarithmic weight, which is however optimal due to the effect of the Hardy potential.
\end{remark}
%
%% \frac{1}{|x|^{d-2}} 
% 
%result corresponds to an extension of this bound to orthogonal functions. We believe that the logarithm in \eqref{eq:pointwise-Bessel} is a consequence of orthogonality, and not of the Hardy potential. In other words, this logarithm is sharp even for the Dirichlet Laplacian instead of $\mathcal{L}_B$ . 

Our proof of Lemma \ref{lem:radial}  uses an analogue of Rumin's method \cite{R11}, plus the precise spectral property of $\mathcal{L}_B$ which has been understood by V\'azquez--ZuaZua \cite{VazZua-00}. Although the result there holds only the unit ball, its application to the whole $\R^d$ is made possible due to the relation
\begin{align} \label{eq:KVW}
N(0, P(\mathcal{L}_{\R^d}-W)P) \le 1 +  N(0, P(\mathcal{L}_{B}-\1_{B}W)P) + N(0, P(\mathcal{L}_{B^c}-\1_{B^c}W)P)
\end{align}
where $P$ is the projection onto radial functions. The bound \eqref{eq:KVW} follows from the fact that imposing the Dirichlet boundary condition at $|x|=1$ in one-dimension is a rank-one perturbation. The same idea was used by Kova\v r\'ik--Vugalter--Weidl \cite{KVWCMP07} to derive CLR type estimates for Schr\"odinger operators in 2D. The conclusion of Theorem \ref{thm:CLR} then follows from an inversion in the unit sphere which allows to control $N(0, P(\mathcal{L}_{B^c}-\1_{B^c}W)P)$ via $N(0, P(\mathcal{L}_{B}-\1_{B}W)P)$.

Note that the Hardy--Lieb--Thirring inequality \eqref{eq:intro-HLT} has been extended to the fractional Laplacian $(-\Delta)^s$ with singular potential $|x|^{-2s}$ \cite{FLSJAMS08,Frank-09}, which in the case $s=1/2$ is relevant to applications in the stability of relativistic matter. Again, the corresponding bound has been  known only for eigenvalue moments $\gamma>0$. The case $\gamma=0$ is the content of our next result. 

Let $d\ge 1$ and  $0<s<\min(1, d/2)$. Let $(-\Delta)^s$ be the fractional Laplacian on $L^2(\R^d)$ defined via the quadratic form 
	\begin{equation}\label{defi-semi-norm}
h_s[u]=  \langle  u, (-\Delta)^s u \rangle = a_{s,d}  \int_{\R^d} \int_{\R^d}  \frac{|u(x) -u(y)|^2}{|x-y|^{d+2s}} {\rm d} x \,  {\rm d} y, \quad {a}_{s,d} =  2^{2s-1} \frac{\Gamma(\frac{d+2s}{2})}{\pi^{\frac{d}{2}}|\Gamma(-s)|}.
\end{equation}
Consider the fractional Hardy--Schr\"odinger operator
\begin{align} \label{eq:fractional-Hardy}  \mathcal{L}_s  =  (-\Delta)^s  -  \frac{ \mathcal{C}_{s, d}}{|x|^{2s}} \ge 0 \quad \text{ on }L^2(\R^d),   \quad  \mathcal{C}_{s,d} = 2^{2s} \frac{\Gamma^2\left(\frac{d + 2s}{4}\right)}{\Gamma^2\left(\frac{d - 2s}{4}\right)},
\end{align}
where $\mathcal{C}_{s, d}$ is the optimal constant in the fractional Hardy inequality \cite{Herbst-77}. We have the following extension of Theorem \ref{thm:CLR} to the fractional case. 

\begin{theorem}[Fractional CLR type bound]\label{thm:fractional} For every dimension $d\ge 1$ and $0<s<\min(1,d/2)$, there exists a constant $C_{s,d}>0$ independent of the real-valued potential $W$ such that

$$ N(0, \mathcal{L}_s  -W)   \le C_{s,d} \left(1  + \int_{\R^d} W(x)_+^{\frac{d}{2s}} (1+|\ln|x||)^{\frac{d-s}{s}}{\rm d} x\right).$$
\end{theorem}

We remark that the same inequality holds in the presence of a magnetic field.

The proof of Theorem \ref{thm:fractional} deviates substantially from that of  Theorem \ref{thm:CLR}. On the one hand, we will split again $\R^d$ into $B$ and $B^c$, and use crucially the fractional Hardy--Sobolev inequality by Tzirakis \cite{TJFA2016} on $B$.  On the other hand, the generalization from the one-body inequality to the many-body inequality on each domain ($B$ or $B^c$) are done via an abstract equivalence of Sobolev and CLR inequalities, a strategy proposed in  \cite{FLSICMP09}.  The conclusion also requires a careful implementation of the localization method in the fractional case.

In fact, the proof of Theorem \ref{thm:fractional} is also valid when $s=1$ and $d\geq 3$ and simplifies at some points. Thus, we obtain an alternative proof of Theorem \ref{thm:CLR}. In the local case, however, we feel that the first proof is more direct, which motivated us to present it first. 

We will prove Theorem \ref{thm:CLR} and Theorem \ref{thm:fractional} in Section \ref{sec:Thm1} and Section \ref{sec:Thm2}, respectively. 

%%%%%%%%%%%%%%%%%%%%%%%%%%%%%%%%%%%%%%%%%%%%%%%%
%%%%%%%%%%%%%%%%%%%%%%%%%%%%%%%%%%%%%%%%%%%%%%%%
%%%%%%%%%%%%%%%%%%%%%%%%%%%%%%%%%%%%%%%%%%%%%%%%
%%%%%%%%%%%%%%%%%%%%%%%%%%%%%%%%%%%%%%%%%%%%%%%%

\section{Local Case}\label{sec:Thm1}

\subsection{Improved Hardy Inequalities} 
We denote by $P$ the projection onto radially symmetric functions in $L^2(\R^d)$ and set $P^\bot =1-P$. On the non-radial part, the following improved Hardy inequality is well-known (see e.g. \cite{So,BiLa,EkFr,FTJFA2002,MJFA09}). 

	\begin{lemma} \label{lem:improved-Hardy-1} We have the operator inequality on $L^2(\R^d)$
	\begin{equation} \label{eq:improved-Hardy-nonradial}
	P^\bot(-\Delta)P^\bot \geq \left( \frac{(d-2)^2}{4} + (d-1) \right) P^\bot |x|^{-2} P^\bot .
\end{equation}
	\end{lemma}
The above estimate comes from the fact that the lowest nontrivial eigenvalue of the Laplace--Beltrami operator on the unit sphere $\Sph^{d-1}$ is equal to $d-1$.
 	
	On the radial part, we recall the following result from Musina \cite{MJFA09}. 
	\begin{lemma}[{\cite[Proposition 1.1]{MJFA09}}] \label{lem:improved-Hardy-2} For every radial function $u\in H_0^1(B)$, we have 
	\begin{equation}	\label{eq:Musina} \begin{aligned}
	\langle u, \cL u\rangle \ge  \frac 1 4 \int_B \frac{|u(x)|^2}{|x|^2 |\ln |x||^2} {\rm d}x. 
	\end{aligned} \end{equation}
	\end{lemma}
	
%%%%%%%%%%%%%%%%%%%%%%%%%%%%%%%%%%%%%%%
%%%%%%%%%%%%%%%%%%%%%%%%%%%%%%%%%%%%%%%
%%%%%%%%%%%%%%%%%%%%%%%%%%%%%%%%%%%%%%%
%%%%%%%%%%%%%%%%%%%%%%%%%%%%%%%%%%%%%%%

\subsection{Strauss type Estimate} In this subsection we restrict to radial functions and prove the pointwise estimate in Lemma \ref{lem:radial}. 

\begin{proof}[Proof of Lemma \ref{lem:radial}] Let $z_{0,k}$ be the $k$-th zero of the Bessel function $J_0$. From the spectral property of $\mathcal{L}_B$ studied in \cite{VazZua-00}, we have
$$P \mathcal{L}_B =\sum_{k\ge 1} \lambda_{0,k} | \varphi_{0,k}\rangle \langle \varphi_{0,k}|$$
with the ($L^2$-normalized) eigenfunctions
$$
\varphi_{0,k} (x) =  \frac{1}{|x|^{\frac {d-2}{2}} \sqrt{|\mathbb{S}^{d-1}|\int_0^1 r J_0^2 (z_{0,k} r) {\rm d} r}} J_0 \left(z_{0,k} |x|\right), \quad x \in B \setminus \{0\}, 
$$
and the corresponding eigenvalues $\lambda_{0,k}=z_{0,k}^2$. Here $|\mathbb{S}^{d-1}|$ is the surface area of the unit sphere in $\R^d$. Consequently, for a.e. $x \in B$,
\begin{align*}
\rho(x) &= \sum_{n \ge 1} | u_n (x)|^2 = \sum_{n \ge 1} \left| \sum_{k\ge 1} \langle \varphi_{0,k}, u_n\rangle \varphi_{0,k}(x) \right|^2 \\
&=  \sum_{n \ge 1} \left| \int_{B}  \sum_{k\ge 1} \frac{1}{\sqrt{\lambda_{0,k}}}  \overline{\varphi_{0,k}(y)}   \varphi_{0,k}(x)  \sqrt{\mathcal{L}_B}u_n(y){\rm d} y \right|^2 
\\
&\le \int_{B}   \left| \sum_{k\ge 1} \frac{1}{\sqrt{\lambda_{0,k}}} \overline{\varphi_{0,k}(y)}   \varphi_{0,k}(x) \right|^2 \,{\rm d} y   \le  \sum_{k\ge 1}  \frac{|\varphi_{0,k}(x)|^2}{\lambda_{0,k}}. 
\end{align*}
Here we used Bessel's inequality via the (sub)orthogonality of $\{\sqrt{\mathcal L}_B u_n\}_n$, or more precisely the condition \eqref{eq:Pauli}. Using the asymptotic properties of Bessel function (see \cite[p. 199]{Wat1944}) 
$$J_0(0)=1,\quad J_0(r)= \sqrt{\frac{2}{\pi r}} \cos\left(r- \frac \pi 4\right) + O(r^{-1})_{r\to \infty},$$
we have
$$
J_0^2(r) \lesssim \min (1, r^{-1}), \quad z_{0,k} \sim k, \quad \lambda_{0,k} \int_0^1 r J_0^2 (z_{0,k} r) {\rm d} r = \int_0^{z_{0,k}} r J_0^2(r) {\rm d} r \sim k.
$$
Hence, 
\begin{align*}
\rho(x)  \le \sum_{k\ge 1}  \frac{|\varphi_{0,k}(x)|^2}{\lambda_{0,k}}  \lesssim_d \frac{1}{|x|^{d-2}} \sum_{k\ge 1} \frac{1}{k} \min\left(1,\frac{1}{k|x|}\right) \lesssim_d  \frac{1}{|x|^{d-2}}  (1+|\ln |x||), 
\end{align*}
which is the desired pointwise estimate \eqref{eq:pointwise-Bessel}. 
	\end{proof}	
	
	\begin{remark}[Laplacian case]\label{rmk:optimality} If radial functions $\{v_n\}_{n\ge 1} \subset H_0^1(B)$ satisfy the orthogonality \eqref{eq:Pauli} with  $\mathcal{L}_B$ replaced by the Dirichlet Laplacian $-\Delta_B$ in dimensions $d\ge 3$, then following the above proof of Lemma \ref{lem:radial} and using the well-known spectral properties of $-\Delta_B$ (see, e.g., \cite[Sec 6.4.4]{Triebel1992}) 	we find that
	$$
	\sum_{n\ge 1} |v_n(x)|^2 \le \sum_{k\ge 1} \frac{|\varphi_k(x)|^2}{\lambda_k}
	$$
where 
	$$
	\varphi_k (x) = \frac{1}{|x|^{\frac {d-2}2} \sqrt{ |\mathbb{S}^{d-1}|\int_0^1 r J^2_{\frac {d-2}2 }(z_k r) {\rm d} r} } J_{\frac {d-2}2 }(z_k |x|)
	$$
	are ($L^2$-normalized)  radial eigenfunctions of $-\Delta_B$ with eigenvalues $\lambda_k=z_k^2$, and $z_k$ is the $k$-th zero of $J_{\frac {d-2}{2}}$. From the asymptotic properties of Bessel function (see \cite[p. 40, 199]{Wat1944})
	$$
	J_{\frac {d-2}2 }(r) \sim_{r \to 0} \frac{1}{\Gamma(d/2)} \left( \frac r 2 \right)^{\frac {d-2}{2}}, \quad  J_{\frac {d-2}2 }(r) \sim_{r\to \infty} \sqrt{\frac{2}{\pi r}} \cos\left(r- \frac{(d-1)\pi}{4}\right) + O(r^{-1}), 
	$$
	we obtain
	$$
	J^2_{\frac{d-2}{2}}(r) \lesim \min (r^{d-2}, r^{-1}),\quad z_k \sim k, \quad \lambda_k \int_0^1 r J^2_{\frac {d-2}2 }(\sqrt{\lambda_k} r) {\rm d} r = \int_0^{\sqrt{\lambda_k}} r J^2_{\frac {d-2}2 }(r) {\rm d} r \sim \sqrt{\lambda_k}, 
	$$
and hence
\begin{align} \label{eq:Laplacian-Strauss-density}
	\sum_{n\ge 1} |v_n(x)|^2 \lesim_d \frac{1}{|x|^{d-2}} \sum_{ k \ge 1}  \frac{J^2_{\frac {d-2}2}(\sqrt{\lambda_k} |x|)}{ \sqrt{\lambda_k} } \lesim_d \frac{1}{|x|^{d-2}} \sum_{k \ge 1} \frac{1}{k} \min \left( (k|x|)^{d-2},  \frac{1}{k|x|} \right) \lesim_d  \frac{1}{|x|^{d-2}}.
\end{align}
The bound \eqref{eq:Laplacian-Strauss-density} is slightly better than \eqref{eq:pointwise-Bessel} as it does not contain a logarithmic weight. 
\end{remark}
	
%%%%%%%%%%%%%%%%%%%%%%%%%%%%%%%%%%%%%%%
%%%%%%%%%%%%%%%%%%%%%%%%%%%%%%%%%%%%%%%
%%%%%%%%%%%%%%%%%%%%%%%%%%%%%%%%%%%%%%%
%%%%%%%%%%%%%%%%%%%%%%%%%%%%%%%%%%%%%%%

\subsection{Conclusion of Theorem \ref{thm:CLR}}

Since $\cL-W \ge \cL - W_+$, by the min-max principle it suffices to assume that $W\ge 0$. Recall that $P$ is the projection onto radially symmetric functions in $L^2(\R^d)$ and $P^\bot =1-P$. Then, using $W\geq 0$, we have the Cauchy--Schwarz inequality
$$
W \le 2 P  W P + 2 P^\bot W P^\bot,
$$
and hence
\begin{align}\label{eq:split-radial-nonradial}
N(0,\mathcal L - W) \le N(0,P(\mathcal L - 2W) P) + N(0,P^\bot (\mathcal L - 2W) P^\bot). 
\end{align}

On the non-radial part, using the improved Hardy inequality in Lemma \ref{lem:improved-Hardy-1} we have
	$$
	P^\bot \mathcal L P^\bot \gesim_d P^\bot  (-\Delta)P^\bot.
	$$
	Therefore, by the standard CLR inequality,  
	\begin{align} \label{eq:CLR-nonradial}
	N(0,P^\bot (\mathcal L - 2W) P^\bot) \le N(0,P^\bot (C_d^{-1} (-\Delta) - 2W) P^\bot) \lesssim_d  \int_{\R^d} W(x)^{\frac d 2} {\rm d} x. 
	\end{align}
	
	On the radial part, using the Hoffmann--Ostenhof inequality \cite{HO77}, the  improved Hardy inequality in Lemma \ref{lem:improved-Hardy-2}, and the pointwise estimate in Lemma \ref{lem:radial}, we get
	\begin{equation}	\label{eq:Musina} \begin{aligned}
	\sum_{n \ge 1} \| \sqrt{\mathcal L}_B u_n\|_{L^2(B)}^2 &\ge \| \sqrt{\mathcal L}_B \sqrt{\rho}\|_{L^2(B)}^2 \\
	&\ge \frac 1 4 \int_B \frac{\rho(x)}{|x|^2 |\ln |x||^2} {\rm d}x \gtrsim_d \int_{B} \frac{\rho(x)^{\frac{d}{d-2}}}{(1+ |\ln |x||)^{1+\frac{d}{d-2}}} {\rm d}x
	\end{aligned} \end{equation}
	for all radial functions  $\{u_n\}_{n\ge 1}$ satisfying \eqref{eq:Pauli}. By a standard duality argument (see e.g. \cite{Frank14}), the kinetic inequality \eqref{eq:Musina} implies that 
	\begin{align}\label{eq:CLR-radial-ball}
	N(0, P(\mathcal{L}_B - 2\1_{B}W)P) \lesssim_d \int_{B} W(x)_+^{\frac d2} (1+|\ln |x||)^{d-1} {\rm d} x.  
	\end{align}	
	
Next, we use \eqref{eq:KVW}, namely 
\begin{align} \label{eq:KVW-app}
N(0, P(\mathcal{L}_{\R^d}-2W)P) \le 1 +  N(0, P(\mathcal{L}_{B}-2\1_{B}W)P) + N(0, P(\mathcal{L}_{B^c}-2\1_{B^c}W)P).
\end{align}
To control $N(0, P(\mathcal{L}_{B^c}-2\1_{B^c}W)P)$, we use an inversion in the unit sphere. Let us introduce some notation. Let $\mathcal Q$ be the form domain of the operator 
$$-r^{1-d}\partial_r \left( r^{d-1}\partial_r \right) - \frac{(d-2)^2}{4r^2}\quad \text{in }L^2((1,\infty),r^{d-1}dr)$$
with a Dirichlet boundary condition at $r=1$. Similarly, let $\tilde{\mathcal Q}$ be the form domain of the operator 
$$-s^{1-d}\partial_s \left( s^{d-1}\partial_s \right) - \frac{(d-2)^2}{4s^2}\quad \text{in }L^2((0,1),s^{d-1}ds)$$
with a Dirichlet boundary condition at $s=1$ (and a `Dirichlet boundary' condition at $s=0$).

\begin{lemma}\label{inversion}
	Assume that $u\in\mathcal Q$ and $\tilde u\in\tilde{\mathcal Q}$ are related by
	$$
	u(r) = r^{2-d} \tilde u(1/r)
	\qquad \text{for all}\ r\in (1,\infty) \,.
	$$
	Then
	$$
	\int_1^\infty \Big(u'(r)^2-\frac{(d-2)^2}{4r^2} u(r)^2\Big) r^{d-1}\,{\rm d}r = \int_0^1 \Big(\tilde u'(s)^2-\frac{(d-2)^2}{4s^2} \tilde u(s)^2 \Big) s^{d-1}\,{\rm d}s \,.
	$$
\end{lemma}

\begin{proof}
	By an approximation argument, we may assume the $u\in C^2_c(1,\infty)$ and $\tilde u\in C^2_c(0,1)$. Then the assertion follows by a straightforward computation, which we omit.
\end{proof}

\begin{corollary}\label{inversioncor}
	Assume that $W$ defined on $\overline{B}^c$ and $\tilde W$ defined on $B$ are related by
	$$
	W(x) = |x|^{-4} \tilde W(x/|x|^2)
	\qquad\text{for all}\ x\in \overline{B}^c \,.
	$$
	Then
	$$
	N(0,P(\mathcal L_{B^c}-\1_{B^c}W)P) = N(0,P(\mathcal L_{B}-\1_{B}\tilde W)P) \,.
	$$
\end{corollary}

\begin{proof}
	Clearly, the assertion only depends on the spherical means of $W$ and $\tilde W$, which we denote by $w$ and $\tilde w$. By Glazman's lemma (see e.g. \cite[Theorem 1.25]{FLW2023}), we have
	\begin{align*}
		& N(0,P(\mathcal L_{B^c}-\1_{B^c}W)P) \\
		& = \sup\Big\{ \dim M:\ \int_1^\infty \Big(u'(r)^2-\frac{(d-2)^2}{4r^2} u(r)^2\Big) r^{d-1}\,{\rm d}r < \int_1^\infty w(r) u(r)^2 r^{d-1}\,{\rm d}r \ \forall 0 \neq u\in M \Big\}
	\end{align*}
	and
	\begin{align*}
		& N(0,P(\mathcal L_{B}-\1_{B}\tilde W)P) \\
		& = \sup \Big\{ \dim \tilde M:\ \int_0^1 \Big(\tilde u'(s)^2-\frac{(d-2)^2}{4s^2} \tilde u(s)^2 \Big) s^{d-1}\,{\rm d}s < \int_0^1 \tilde w(s) \tilde u(s)^2 s^{d-1}\,{\rm d}s \ \forall 0 \neq \tilde u \in\tilde M \Big\} 
	\end{align*}
	where $M$ and $\tilde M$ run through subspaces in $\mathcal Q$ and $\tilde{\mathcal Q}$, respectively. The claimed equality therefore follows from the identity in the lemma, the identity
	$$
	\int_1^\infty w(r) u(r)^2 r^{d-1}\,{\rm d}r = \int_0^1 \tilde w(s) \tilde u(s)^2 s^{d-1}\,{\rm d}s, 
	$$
	together with the fact that the correspondence $u\mapsto \tilde u$ is bijective on form cores consisting of functions vanishing near the origin and near infinity, respectively.
\end{proof}

It is now easy to finish the proof of Theorem \ref{thm:CLR}. Clearly, we have, under the conditions of Corollary \ref{inversioncor},  
$$
\int_{B^c} W(x)^\frac d2 (1+|\ln |x||)^{d-1}\,{\rm d}x = \int_{B} \tilde W(x)^\frac d2 (1+|\ln |x||)^{d-1}\,{\rm d}x \,.
$$
This identity, together with Corollary \ref{inversioncor} and inequality \eqref{eq:CLR-radial-ball}, yields
\begin{equation}
	\label{eq:CLR-radial-balloutside}
	N(0,P(\mathcal L_{B^c}-2\1_{B^c}W)P) \lesssim_d \int_{B^c} W(x)^\frac d2 (1+|\ln |x||)^{d-1}\,{\rm d}x \,.
\end{equation}
Therefore, inserting \eqref{eq:CLR-radial-ball} and \eqref{eq:CLR-radial-balloutside} into \eqref{eq:KVW-app}, we obtain
$$
N(0, P(\mathcal{L}_{\R^d}-2W)P) \le 1 + C_d \int_{\R^d} W(x)_+^{\frac d2} (1+|\ln |x||)^{d-1} {\rm d}x\,.
$$
This, together with \eqref{eq:split-radial-nonradial} and \eqref{eq:CLR-nonradial}, completes the proof of Theorem \ref{thm:CLR}. 
\qed

%%%%%%%%%%%%%%%%%%%%%%%%%%%%%%%%%%%%%%%%%%%%%%%%
%%%%%%%%%%%%%%%%%%%%%%%%%%%%%%%%%%%%%%%%%%%%%%%%
%%%%%%%%%%%%%%%%%%%%%%%%%%%%%%%%%%%%%%%%%%%%%%%%
%%%%%%%%%%%%%%%%%%%%%%%%%%%%%%%%%%%%%%%%%%%%%%%%

\section{Nonlocal Case}\label{sec:Thm2}

\subsection{Fractional Hardy--Sobolev Inequalities} Let $0<s<\min(1,d/2)$ and let $\cL_s$ be defined in \eqref{eq:fractional-Hardy}.  Recall the following results of Tzirakis \cite{TJFA2016}. 

\begin{lemma}[{\cite[Theorem 3 and Theorem 5]{TJFA2016}}]\label{lem:fractional-HS-B}
For all $ v \in C^{\infty}_c(B)$, we have
\begin{equation*}
	\langle v,\cL_sv\rangle  \gesim_{s,d}  \left(\int_{B}    |v(x)|^{\frac{2 d}{d-2 s}} ( 1+|\ln |x||)^{-\frac{2(d-s)}{d-2s}}\mathrm{~d} x\right)^{\frac{d-2 s}{d}} ,
\end{equation*}
and 
\begin{equation*}
	\langle v,\cL_sv\rangle  \gesim_{s,d}    \int_{B} \frac{|v(x)|^{2} }{(1+ |\ln|x||^2) |x|^{2 s}} \mathrm{~d} x.
\end{equation*}
\end{lemma}

%%%%%%%%%%%%%%%%%%%%%%%%%%%%%%%%%%%%%%%%
%%%%%%%%%%%%%%%%%%%%%%%%%%%%%%%%%%%%%%%%

\subsection{Equivalence of Sobolev and CLR Inequalities}

In this part, we recall the equivalence of Sobolev and CLR inequalities from \cite{FLSICMP09}. 
Let $X$ be a separable measure space. We consider the measure on $X$ as fixed and denote integration with respect to this measure by ${\rm d} x$. By $L^{p}(X)$ for $1 \leq p \leq \infty$ we denote the usual $L^{p}$ space with respect to this measure.

Let $t$ be a non-negative quadratic form with domain dom $t$ that is closed in the Hilbert space $L^{2}(X)$ and let $T$ be the corresponding self-adjoint operator.

Throughout this paper we work under the following assumption, which depends on a parameter $1<\kappa<\infty$.

\begin{assumption}[Generalized Beurling--Deny conditions]\label{ass:bd}
	$\phantom{x}$
	\vspace{-.1truecm}
	\begin{enumerate}
		\item[(a)] 
		if $u,v\in\dom t$ are real-valued, then $t[u+iv]=t[u]+t[v]$,
		\item[(b)]
		if $u\in\dom t$ is real-valued, then $|u|\in\dom t$ and $t[|u|]\leq t[u]$.
		\item[(c)]
		there is a measurable, a.e.\ positive function $\mu$ such that, if $u\in\dom t$ is non-negative then $\min(u,\mu)\in\dom t$ and $t[\min(u,\mu)]\leq t[u]$. Moreover, there is a form core $\mathcal Q$ of $t$ such that $\mu^{-1}\mathcal Q$ is dense in $L^2(X, \mu^{2\kappa/(\kappa-1)} {\rm d} x)$.
	\end{enumerate}
\end{assumption}

\begin{theorem}[Equivalence of Sobolev and CLR inequalities in the presence of weights]\label{thm:clrweighted}
	Let Assumption \ref{ass:bd} be satisfied for some $\kappa>1$ and let $w$ be a nonnegative, measurable function on $X$ that is finite a.e. Then the following are equivalent:
	\begin{enumerate}
		\item[(i)] $T$ satisfies a weighted Sobolev inequality with exponent $q=2\kappa/(\kappa-1)$, that is, there is a constant $S>0$ such that for all $u\in\dom t$,
		\begin{equation}\label{eq:sobolevweighted}
			t[u] \geq S \left( \int_X |u|^q w^{-(q-2)/2}\,{\rm d} x \right)^{2/q} \,.
		\end{equation}
		\item[(ii)] $T$ satisfies a weighted CLR inequality with exponent $\kappa$, that is, there is a constant $L>0$ such that for all $0\leq V\in L^\kappa(X,w\,{\rm d} x)$,
		\begin{equation}\label{eq:thm:clrweighted}
			N(0,T-V) \leq L \int_X V^{\kappa} w \,{\rm d} x \,.
		\end{equation}
	\end{enumerate}
	The respective constants are related according to
	\begin{equation}\label{eq:CLRconst}
		S^{-\kappa} \leq L  \leq e^{\kappa-1} S^{-\kappa} \,.
	\end{equation}
\end{theorem}

\begin{proof}
	This theorem for $w=1$ appears in the paper  \cite{FLSICMP09}. It is based on a method due to Li and Yau \cite{LiYa83} with an improvement in \cite{BlStRe} and generalizes a theorem of Levin and Solomyak \cite{LeSo}.
	
	We now prove the result for general $w$ as in the statement of the theorem. The implication (ii$\implies$i) follows by a standard application of H\"older's inequality. Thus, we only need to prove (i$\implies$ii). For the proof, we note that we may assume that $w$ is bounded away from zero. Indeed, once the implication is proved under this extra assumption, we can apply it to $w_\epsilon:= w+\epsilon$ in place of $w$. This $w_\epsilon$ still satisfies \eqref{eq:sobolevweighted} and it satisfies the extra condition $w_\epsilon\geq \epsilon$. Thus, we obtain \eqref{eq:thm:clrweighted} with $w_\epsilon$ in place of $w$. Since the constant is independent of $\epsilon$, we can let $\epsilon\to 0$ and obtain, by monotone convergence, the claimed inequality \eqref{eq:thm:clrweighted} with $w$.
	
	Thus, assume that $w$ is bounded away from zero. To better explain the strategy of the proof, let us first assume, in addition, that $w\in L^\infty(X)$. Then the Hilbert space $\mathfrak H := L^2(X,w^{-(q-2)/2}{\rm d} x)$ coincides, with equivalent norm, with the Hilbert space $L^2(X)$. We consider the quadratic form $t$ in the Hilbert space $\mathfrak H$. It is clearly nonnegative and, by our assumptions on $w$, it is closed. (It is in the proof of closedness that we use the assumption $w\in L^\infty(X)$.) Thus it generates a nonnegative operator $A$ in $\mathfrak H$. Moreover, it satisfies Assumption \ref{ass:bd}, since we are assuming that the corresponding assumption is satisfied for the original form $t$ in the Hilbert space $L^2(X)$. Applying \cite[Theorem 2.1]{FLSICMP09} to the operator $A$, we obtain that for any $0\leq U \in L^\kappa(X,w^{-(q-2)/2}{\rm d} x)$ one has
	$$
	N(0,A-U) \leq e^{\kappa-1} S^{-\kappa} \int_X U^\kappa \,w^{-\frac{q-2}{2}}\,{\rm d} x \,.
	$$
	At the point we notice that the quadratic form of the operator $A-U$ is
	$$
	t[v] - \int_X U |v|^2 w^{-\frac{q-2}{2}}\,{\rm d} x \,.
	$$
	Thus, by Glazman's lemma (see e.g. \cite[Theorem 1.25]{FLW2023}),
	$$
	N(0,A-U) = \sup\Big\{\! \dim M:\ t[v] - \int_X U |v|^2 w^{-\frac{q-2}{2}}\,{\rm d} x <0, \; \forall 0\neq v\in M \Big\} = N(0,T - Uw^{-\frac{q-2}{2}} ) \,,
	$$
	where we emphasize that the operator $A-U$ acts in $\mathfrak H$ and the operator $T - Uw^{-\frac{q-2}{2}}$ acts in $L^2(X,{\rm d} x)$. (We emphasize that the crucial point in this `trick' is that Glazman's lemma only sees the quadratic form, but not the norm in the Hilbert space.) Writing $V=Uw^{-\frac{q-2}{2}}$ we have shown that
	$$
	N(0,T-V) \leq e^{\kappa-1} S^{-\kappa} \int_X V^\kappa w\,{\rm d} x \,.
	$$
	This is the assertion (ii) that we wanted to prove under the assumption $w\in L^\infty(X)$.
	
	Now we consider the general case, where $w$ is finite a.e. As we mentioned before, we may also assume that $w$ is bound away from zero. The Hilbert space $\mathfrak H$ is defined as before, but now we only know that $L^2(X)\subset\mathfrak H$. Given $\delta>0$ we consider the quadratic form $v\mapsto t_\delta[v] := t[v] + \delta \int_X |v|^2\,{\rm d} x$ in the Hilbert space $\mathfrak H$ with form domain $\dom t$. Let us show that $t_\delta$ is closed in $\mathfrak H$. Thus, let $(v_j)\subset\dom t$ and $v\in\mathfrak H$ such that $v_j\to v$ in $\mathfrak H$ and such that $(v_j)$ is Cauchy with respect to $t_\delta$. Since $t$ is nonnegative, we infer that $(v_j)$ is Cauchy with respect to the norm in $L^2(X)$ and hence convergent in $L^2(X)$. Passing to a.e.-convergent subsequences and recalling that $w$ is finite a.e., it is easy to see that $v\in L^2(X)$ and that $v_j\to v$ in $L^2(X)$. Therefore the assumed closedness of $t$ implies that $v\in\dom t$ and that $t[v_j-v]\to 0$. Thus, we have shown that $t_\delta$ is closed in $\mathfrak H$.
	
	Let us denote by $A_\delta$ the nonnegative operator in $\mathfrak H$ generated by $t_\delta$. The form $t_\delta$ satisfies Assumptions \ref{ass:bd} and we deduce, as before, that
	$$
	N(0,A_\delta -U) \leq e^{\kappa-1} S^{-\kappa} \int_X U^\kappa \,w^{-\frac{q-2}{2}}\,{\rm d} x \,.
	$$
	Meanwhile, again by Glazman's lemma,
	$$
	N(0,A_\delta -U) = N(0,T+\delta - U w^{-\frac{q-2}{2}}) = N(-\delta, T-U w^{-\frac{q-2}{2}}) \,.
	$$
	Thus, we conclude that for $V\geq 0$ satisfying $\int_X V^\kappa w\,{\rm d} x<\infty$ the spectrum of $T-V$ in the interval $(-\infty,-\delta)$ is finite and
	$$
	N(-\delta,T-V) \leq e^{\kappa-1} S^{-\kappa} \int_X V^\kappa w\,{\rm d} x \,.
	$$
	Letting $\delta\searrow 0$, we obtain the claimed inequality.
\end{proof}

We return now to our operator $\mathcal L_s$ and prove the corresponding CLR inequality in the ball $B$. 

\begin{lemma}\label{lem:fractional-CLR-B} Let $W\ge 0$. Consider the operator $\mathcal{L}_s-W$ on $L^2(B)$ defined by the quadratic form restricted to functions vanishing outside of $B$. We have 
	$$
	N(0,\mathcal L_s - W) \lesssim_{s,d} \int_B W(x)^{\frac d{2s}} (1+ |\ln |x||)^{\frac{d-s}{s}}\,{\rm d} x \,.
	$$
\end{lemma}

\begin{proof}
	We apply Theorem \ref{thm:clrweighted}  in the measure space $X=B$ with Lebesgue measure and with $t[u]=\|\mathcal L_s^{1/2} u\|^2$. Items (a) and (b) in Assumption \ref{ass:bd} are clearly satisfied. The first part of item (c) is satisfied with $\mu(x)=|x|^{-(d-2s)/2}$, as follows from the ground state substitution formula in the  \cite{FLSICMP09} on Hardy--Lieb--Thirring. The second part is satisfied, for any $1<\kappa<\infty$, since $C^\infty_c(B\setminus\{0\})$ is a form core of $\mathcal L_s$. Thanks to Lemma \ref{lem:fractional-HS-B}, the Sobolev inequality \eqref{eq:sobolevweighted} is satisfied with $q=\frac{2d}{d-2s}$ and $w=(1+ |\ln |x| |)^\frac{d-s}{s}$. The assertion follows now from the CLR inequality \eqref{eq:thm:clrweighted}.
\end{proof}

Similarly as in the previous section, one can use the inversion in the unit sphere to derive a corresponding inequality in the complement of the unit ball.

\begin{lemma}\label{lem:fractional-CLR-Bc} Let $W\ge 0$. For the operator $\mathcal{L}_s-W$ defined on $L^2(\overline{B}^c)$ with Dirichlet boundary conditions, we have
	$$
	N(0,\mathcal L_s - W) \lesssim_{s,d} \int_{B^c} W(x)^{\frac d{2s}} (1+ |\ln |x||)^{\frac{d-s}{s}}\,{\rm d} x \,.
	$$
\end{lemma}

\begin{proof}
	We keep the notation $\mathcal L_s$ for the operator in $L^2(\overline{B}^c)$ that appears in the lemma and denote the one in $L^2(B)$ from Lemma \ref{lem:fractional-CLR-B} temporarily by $\tilde{\mathcal L}_s$. Assume that functions $u$ and $\tilde u$ in the form domains of $\mathcal L_s$ and $\tilde{\mathcal L}_s$ are related by
	$$
	u(x) = |x|^{2s-d} \tilde u (x/|x|^2)
	\qquad\text{for all}\ x\in\R^d \setminus \{0\}.
	$$
	We claim that
	\begin{equation}
		\label{eq:inversionfrac}
		\langle u, \mathcal L_s u \rangle = \langle \tilde u, \tilde{\mathcal L}_s \tilde u \rangle \,.
	\end{equation}
	Since $C^2_c(B\setminus\{0\})$ and $C^2_c(\overline B^c)$ are form cores for the relative operators, it suffices to assume that $u$ and $\tilde u$ belong to these sets. For such functions it is well known that
	$$
	\langle u,(-\Delta)^s u \rangle = \langle \tilde u , (-\Delta)^s \tilde u \rangle \,.
	$$
	This appears, in a dual form, for instance in \cite{Li}; see also  \cite[Lemma 2.2]{FWJFA2012}. Additionally, by changing variable $y =  \frac{x}{|x|^2}$, we have
	\begin{align*}
		\int_{B^c} \frac{|u(x)|^2}{|x|^{2s}} {\rm d}x  =  \int_{B^c} \frac{|x|^{2(2s-d)} |\tilde u \left(\frac{x}{|x|^2} \right)|^2}{|x|^{2s}} {\rm d}x = \int_B      \frac{|\tilde u(y)|^2}{|y|^{2s}} {\rm d}y \, ,  \end{align*}
	thus proving the claimed identity \eqref{eq:inversionfrac}.
	
	Identity \eqref{eq:inversionfrac} is the analogue of Lemma \ref{inversion}. As in the proof of Corollary \ref{inversioncor}, we deduce that
	$$
	N(0,\mathcal L_s - W) = N(0,\tilde{\mathcal L}_s - \tilde W) \,,
	$$
	where
	$$
	\tilde W(y) := |y|^{-4s} W(y/|y|^2)
	\qquad\text{for all}\ y\in B \backslash\{0\} \,.
	$$
	Since
	$$
	\int_B \tilde W(y)^\frac d{2s} (1+|\ln|y||)^\frac{d-s}{s}\,{\rm d}y = \int_{B^c} W(x)^{\frac d{2s}} (1+ |\ln |x||)^{\frac{d-s}{s}}\,{\rm d} x \,,
	$$
	we deduce the inequality in the lemma from that in Lemma \ref{lem:fractional-CLR-B}.
\end{proof}

\begin{remark}\label{lem:fractional-HS-Bc}
	We record another use of the inversion method employed in the previous proof, which will be useful later. Namely, we have the bound
	\begin{equation*}
		\langle   u, (-\Delta )^s u   \rangle   \ge      	\mathcal{C}_{s, d} \int_{B^c} \frac{|u(x)|^{2}}{|x|^{2 s}} \mathrm{~d} x+ C_{s,d} \int_{B^c} \frac{|u(x)|^{2} }{(|\ln|x||^2+1) |x|^{2 s}} \mathrm{~d} x.
	\end{equation*}
	Indeed, this follows from the second bound in Lemma \ref{lem:fractional-HS-B} applied to the function $v=\tilde u$ from the previous proof.
\end{remark}

%%%%%%%%%%%%%%%%%%%%%%%%%%%%%%%%%%%%%%%%
%%%%%%%%%%%%%%%%%%%%%%%%%%%%%%%%%%%%%%%%

\subsection{Localization}  Consider smooth partition functions $\chi, \eta:\R^d \to [0,1]$ satisfying 
 \begin{equation}\label{eq:property-chi-eta}
  \chi^2(x ) + \eta^2(x) \equiv 1, \quad \supp \chi \subset   \{|x| \le 2\}, \quad \supp \eta \subset   \{ |x| \ge 1\}.
 \end{equation}
 
The following localization estimate for the fractional Laplacian is of independent interest. 
 
		\begin{lemma}\label{propo-L-s-dominate} Let $0<s<\min(1,d/2)$ and $\delta  \in (0,1)$. Then for every $\delta>0$, there exists $C=C(s,d,\delta)>0$ such that
$$
\mathcal{L}_s  \ge \chi ( \mathcal{L}_s - C)  \chi + (1 - \delta)  \eta ( \mathcal{L}_s - C \1_{B_3} )  \eta,
$$
where $\1_{B_3}$ is the indicator function of $B_3=\{|x| < 3\}\subset \R^d$. 
\end{lemma}
		
		\begin{proof} For every $u \in C^\infty_c (\R^d)$, we have the IMS formula
	\begin{equation} \label{eq:IMS-decomposition}
		\langle  u, \mathcal{L}_s u \rangle = \langle u, (\chi\mathcal{L}_s  \chi) u  \rangle + \langle u, (\eta \mathcal{L}_s \eta) u \rangle -  \langle u, \mathcal{H} u\rangle
	\end{equation}
	where $\mathcal{H}$ is  the bounded operator on $L^2(\R^d)$ with integral kernel
	$$ H(x,y) = a_{s,d}   \frac{  (\chi(x) - \chi(y))^2 + (\eta(x) - \eta(y))^2 }{\left| x-y\right|^{d+2s} } .$$
	This formula is due to Michael Loss and appeared in \cite{LiYa}.

By the triangle inequality we have the pointwise estimate 
	$$
	0\le H(x,y){\lesssim_{s,d}}  \frac{ \1_{B_3}(x) \1_{B_3} (y) }{|x-y|^{d+2s-2}} +  \frac{ \1_{B_3}(x)  \1_{B_3^c}(y)}{(1+ |y|)^{d+2s}} +  \frac{\1_{B_3}(y)  \1_{B_3^c}(x)}{(1+ |x|)^{d+2s}} \quad \text{for } x \neq y.
	$$	
	When $d\geq 2$, combining with the Hardy--Littlewood--Sobolev (HLS) inequality and the H\"older inequality, we get
	\begin{align*}
	\left\langle u, \mathcal{H} u\right\rangle &\lesssim_{s,d} \int_{B_3}\int_{B_3} \frac{|u(x)||u(y)|}{|x-y|^{d+2s-2}} {\rm d} x \, {\rm d} y + \int_{B_3}\int_{B_3^c} \frac{|u(x)||u(y)|}{(1+ |y|)^{d+2s}} {\rm d} x \, {\rm d} y \\
	&\lesssim_{s,d} \| \1_{B_3} u \|_{L^{\frac{2d}{d-2s+2}} (\R^d) }^{2} + {\| \1_{B_3} u \|_{L^{1} (\R^d) } } \left\|  \frac{\1_{B_3^c}(y) u(y)}{|y|^s ( 1 + |\ln |y||)} \right\|_{L^2(\R^d)} \left\| \frac{|y|^s ( 1 + |\ln |y||)} { (1+|y|)^{d+2s}}\right\|_{L^2(\R^d)} \\
	&\lesssim_{s,d}  \delta \left\|  \frac{\1_{B_3^c}(y) u(y)}{|y|^s ( 1 + |\ln |y||)} \right\|_{L^2(\R^d)}^2 + (1+\delta^{-1})  \| \1_{B_3} u \|_{L^{2} (\R^d) }^{2}
	\end{align*} 
	for all $\delta>0$. In dimension $d=1$, the exponent $d+2s-2$ is negative for $s< \frac 1 2$, and hence instead of the HLS inequality we can use $|x-y|^{-(d+2s-2)}\lesssim 1$ for $x,y\in B_3$, leading to the same final estimate. 	

By the Hardy inequality with remainder term in Remark \ref{lem:fractional-HS-Bc},
$$
	\left\|  \frac{\1_{B_3^c}(y) u(y)}{|y|^s ( 1 + |\ln |y||)} \right\|_{L^2(\R^d)}^2  \le \left\|  \frac{\eta(y) u(y)}{|y|^s ( 1 + |\ln |y||)} \right\|_{L^2(\R^d)}^2 \lesssim_{s,d}  \langle u,(\eta \mathcal{L}_s  \eta) u \rangle. 
$$
Thus in summary, for every $\delta\in (0,1)$ we have 
$$
	\left \langle u, {\mathcal H} u \right\rangle \le \delta  \langle (\eta \mathcal{L}_s  \eta) u,u \rangle + C_{s,d,\delta}  \| \1_{B_3} u \|_{L^{2} (\R^d) }^{2}.
$$
The conclusion follows by inserting the latter bound in  \eqref{eq:IMS-decomposition}. 
\end{proof}
		
%%%%%%%%%%%%%%%%%%%%%%%%%%%%%%%%%%%%%%%%
%%%%%%%%%%%%%%%%%%%%%%%%%%%%%%%%%%%%%%%%

\subsection{Conclusion of Theorem \ref{thm:fractional}} It suffices to assume that $W\ge 0$. Let $\chi^2+\eta^2=1$ as in \eqref{eq:property-chi-eta}. By Lemma \ref{propo-L-s-dominate}, we have the following quadratic form estimate on $L^2(\R^d)$
 \begin{align*}
 	\mathcal{L}_s -W  \ge  \chi  \Big( \mathcal{L}_s  - W  -  C \Big) \chi + \frac{1}{2} \eta \Big( \mathcal{L}_s  -2 W  -  C \1_{B_3} \Big)  \eta . 
 \end{align*}
 Therefore, 
  $$ N(0,\mathcal{L}_s - W ) \le   N(0, \chi  ( \mathcal{L}_s  - W  -  C) \chi  )  + N(0, \eta ( \mathcal{L}_s  - 2W  -  C \1_{B_3})  \eta   ).$$
Using Lemma \ref{lem:fractional-CLR-B} (with $B$ replaced by $B_3$, the result remains true with a possible change of the implicit constant), we have
	$$
	N(0, \chi  ( \mathcal{L}_s  - W  -  C) \chi  )  \lesssim_{s,d} \int_{\R^d} \Big[ \chi^2(x) (W(x) + C) \Big]^{\frac d{2s}} (1+ |\ln |x||)^{\frac{d-s}{s}}\,{\rm d} x \,.
	$$
Similarly, by Lemma \ref{lem:fractional-CLR-Bc}, 
$$
	N(0, \eta  ( \mathcal{L}_s  - 2W  -  C \1_{B_3}) \eta  )  \lesssim_{s,d} 
	\int_{\R^d} \Big[ \eta^2(x) (W(x) + C \1_{B_3}(x)) \Big]^{\frac d{2s}} (1+ |\ln |x||)^{\frac{d-s}{s}}\,{\rm d} x \,.
	$$
Thus we conclude that
\begin{align*}
 	N(0,\mathcal{L}_s - W ) &\lesssim_{s,d} \int_{\R^d} \Big[ \chi^2(x) (W(x) + C) \Big]^{\frac d{2s}} (1+ |\ln |x||)^{\frac{d-s}{s}}\,{\rm d} x \\
	&\quad + \int_{\R^d} \Big[ \eta^2(x) (W(x) + C \1_{B_3}(x)) \Big]^{\frac d{2s}} (1+ |\ln |x||)^{\frac{d-s}{s}}\,{\rm d} x \\
	& \lesssim_{s,d} 1+ \int_{\R^d} W(x)^{\frac d{2s}} (1+ |\ln |x||)^{\frac{d-s}{s}}\,{\rm d} x.
		\end{align*}
The proof of Theorem \ref{thm:fractional} is complete.
\qed

%%%%%%%%%%%%%%%%%%%%%%%%%%%%%%%%%%%%%%%%
%%%%%%%%%%%%%%%%%%%%%%%%%%%%%%%%%%%%%%%%
%%%%%%%%%%%%%%%%%%%%%%%%%%%%%%%%%%%%%%%%
%%%%%%%%%%%%%%%%%%%%%%%%%%%%%%%%%%%%%%%%

\subsection*{Acknowledgements} We thank the referees for helpful suggestions. 
Partial support through the Deutsche For\-schungs\-gemeinschaft (DFG, German Research Foundation) Germany’s Excellence Strategy EXC - 2111 - 390814868 and through  TRR 352 -- Project-ID 470903074 (G.K. Duong, R.L. Frank, P.T. Nam),  through the U.S. National Science Foundation grant DMS-1954995 (R.L. Frank), and through the Czech Science Foundation Project GA22-17403S (T.M.T. Le, P.T. Nam, P.T. Nguyen) is acknowledged.

%%%%%%%%%%%%%%%%%%%%%%%%%%%%%%%%%%%%%%%%
%%%%%%%%%%%%%%%%%%%%%%%%%%%%%%%%%%%%%%%%
%%%%%%%%%%%%%%%%%%%%%%%%%%%%%%%%%%%%%%%%
%%%%%%%%%%%%%%%%%%%%%%%%%%%%%%%%%%%%%%%%

\end{document}